\documentclass[12pt,a4paper]{article}
\usepackage{fullpage}
\usepackage{amssymb,amsmath,amsthm}
\usepackage[mathscr]{eucal}
\usepackage{pstricks,pst-node,pst-text,pst-3d}
\usepackage[round]{natbib}

\usepackage{mathrsfs}
\usepackage{eurosym}
\usepackage{color}
\usepackage{tikz}
\usetikzlibrary{shapes,calc,automata}
\usepackage{graphicx}
\usepackage{pgfplots}
\usepackage{caption}

\theoremstyle{definition}
\newtheorem{definition}{Definition}
\newtheorem{theorem}{Theorem}
\newtheorem{proposition}{Proposition}
\newtheorem{lemma}{Lemma}

\theoremstyle{remark}
\newtheorem{remark}{Remark}

\def \reel {\mathbb{R}}

\def \cG {\mathcal{G}}

\def \1{{\mathbf{1}}}

\def \1{{\mathbf{1}}}

\begin{document}

\title{An interaction index for multichoice games}
\author{Mustapha Ridaoui${}^1$, Michel Grabisch${}^1$, Christophe
  Labreuche${}^2$\\
\normalsize $^{1}$ Paris School of Economics, Universit\'e Paris I - Panth\'eon-Sorbonne, Paris, France  \\
{\normalsize\tt \{mustapha.ridaoui,michel.grabisch\}@univ-paris1.fr}\\
\normalsize $^{2}$ Thales Research \& Technology, Palaiseau, France\\
{\normalsize\tt christophe.labreuche@thalesgroup.com}
}

\date{Version of \today}
\maketitle

\begin{abstract}
Models in Multicriteria Decision Analysis (MCDA) can be analyzed by means of an
importance index and an interaction index for every group of criteria. We
consider first discrete models in MCDA, without further restriction, which
amounts to considering multichoice games, that is, cooperative games with
several levels of participation. We propose and axiomatize an interaction index
for multichoice games. In a second part, we consider the continuous case,
supposing that the continuous model is obtained from a discrete one by means of
the Choquet integral.
\end{abstract}
\section{Introduction}
\label{Sintro}
An important issue in MultiCriteria Decision Analysis (MCDA) is to be able to
analyse and explain a numerical model, obtained by elicitation of preferences of
the decision maker. A classical way to do this is to assess the importance of
each criterion (see a general approach to define an importance index in
\citep{rigrla17a}). This description of the model may appear to be sufficient
in the case of simple models, which are additive in essence (e.g., additive utility
models), as it is well known that they imply mutual preferential independence of
criteria \citep{kera76}. However, in case of more complex models, the
preferential independence among criteria does not hold any more, and {\it
  interaction} appears among criteria, so that a description of the model by the
sole importance indices is not sufficient any more. For example, for models where aggregation
of preference is done through a Choquet integral w.r.t. a capacity, an
interaction index is defined for any group of criteria \citep{grla10}, which is a generalization of the interaction index
for pairs of criteria proposed by \cite{muso93}.  Roughly
speaking, a positive interaction index induces a conjunctive behavior (like the
minimum operator), while a negative interaction index induces a disjunctive
behavior (maximum). 

The aim of the paper is to propose an axiomatic foundation of an interaction index
for a MCDA model with no special restriction (and in particular, mutual preferential
independence is not supposed to hold). In a first step, the attributes are
supposed to be defined on a finite universe. Then, such a model is equivalent to
what is called a multichoice game in game theory \citep{hsra93}, that is, a game on
a set of players $N$, where each player can play at a level of participation
represented by an integer between 0 and $k$. Up to our
knowledge, there is no definition of an interaction index for multichoice
games. Nevertheless, there exists a general form of interaction index for games on
lattices \citep{grla07b}, and multichoice games with $k$ levels can be considered as
games on the lattice $(k+1)^N$. This interaction index is defined, however, for
any element of the lattice  $x\in (k+1)^N$, i.e., any profile of participation
of the players. This does not make sense for our purpose, since we are looking
for an interaction index defined for {\it groups} of players/criteria. It is the
contribution of this paper to provide such an index, and to give a
characterization of it.

The paper is organized as follows. Section~\ref{Spre} introduces the necessary
material and notation. Section~\ref{SpreW} summarizes previous works on the
interaction index (the case of classical games and the case of games on
lattices). Our work on the importance index for multichoice games is summarized
in Section~\ref{sec:imin}, since some of the axioms are necessary for our
approach. Section~\ref{Saxiom} gives the main result of the paper, which is the
definition and characterization of an interaction index for multichoice games,
and consequently for general discrete MCDA models. In Section~\ref{sec:Cho}, we
address the continuous case, supposing that the model is obtained from a
discrete one via the Choquet integral.

\section{Preliminaries}
\label{Spre}

Throughout the paper, the cardinality of sets will be denoted by corresponding
lower case letters, i.e., $|N|:=n$, $|S|:=s$, etc. For notational convenience,
we will omit braces for singletons, i.e., $S\cup\{i\}$ is written $S\cup i$, etc.

\medskip

Let $N=\lbrace 1, \ldots, n\rbrace$ be a fixed and finite set which can be
thought as the set of attributes or criteria (in MCDA), players (in cooperative game
theory), etc., depending on the domain of application. In this paper, we will mainly focus on MCDA applications. 

\medskip

We suppose that each attribute $i \in N$ takes values in a set $L_i$, which is
supposed to be finite and denoted by $L_i = \{0, 1,\ldots, k_i\}$. The
alternatives are represented as elements of the Cartesian product $L:=L_1\times
\ldots \times L_n$. An alternative is thus written as a vector $x=(x_1,\ldots,
x_n)$ where $x_i\in L_i$ for all $i\in N$. For each $i\in N$, we denote by
$L_{-i}$ the set $\times_{j\neq i}L_j$. For each $y_{-i} \in L_{-i}$, and any
$\ell \in L_i$, $(y_{-i}, \ell_i)$ denotes the compound alternative $x$ such
that $x_i = \ell$ and $x_j = y_j, \forall j \neq i$. The vector $0_N =
(0,\ldots,0)$ is the null alternative of $L$, and $k_N = (k_1,\ldots, k_n)$ is
the top element of $L$. For each $x \in L$, we denote by $S(x)=\{i\in N \mid
x_i>0\}$ the support of $x$, and by $K(x) = \{i \in N \vert x_i = k_i \}$ the
kernel of $x$.

\medskip

Let $x,y\in L$ and $T\subseteq N\setminus \{\emptyset\}$. $x_T$ is the restriction of $x$ to $T$. We write $x\leq y$ if $x_i\leq y_i$ for every $i\in N$, $x_T<k_T$ if $x_T\leq (k-1)_T$ and $x_T > 0_T$ if $x_T \geq 1_T$.

\medskip

The preferences of a Decision Maker (DM) over the alternatives are supposed to be represented by a function $v:L \rightarrow \reel$. For the sake of generality, we do not make any assumption on $v$, except that 
\begin{equation}
\label{Eqnul}
v(0_N) = 0.
\end{equation}

\medskip

For convenience, we assume from now on that all attributes have the same number of elements,
i.e., $k_i=k$ for every $i\in N$ ($k\in \mathbb{N}$).  Note that if this is not
the case, we set $k = \max_{i\in N} k_i$, and we extend $v: L \rightarrow \reel$ to $v': \{0,\ldots,k\}^N \rightarrow \reel$ by 
\[ v'(x) = v(y) \mbox{ where } y_i = \min(x_i,k_i) \ \forall i\in N .
\]
This amounts to duplicating the last element $k_i$ of $L_i$ when $k_i<k$.  Under
this assumption, we recover well-known concepts.

\medskip

When $k = 1$, $v$ is a \textit{pseudo-Boolean function} $v:\{0,1\}^N\rightarrow\reel$ vanishing at $0_N$.  It can be put in the form of a function $\mu:2^N \rightarrow \reel$, with $v(\emptyset)=0$, which is a \textit{game} in cooperative game theory.  A \textit{capacity} \citep{cho53} or \textit{fuzzy measure} \citep{sug74} is a monotone game, i.e., satisfying $v(A)\leq v(B)$ whenever $A\subseteq B$. For the general case (when $k\geq 1$), $v: L \rightarrow \reel$ fulfilling \ref{Eqnul} corresponds exactly to the concept of {\it multichoice game} \citep{hsra93}, and the numbers $0,1,\ldots,k$ in $L_i$ are seen as the level of activity of the players. A \textit{$k$-ary capacity} \citep{grla03b} is a multichoice game $v$ satisfying the monotonicity condition: or each $x, y \in L$ s.t. $x\leq y$, $v(x)\leq v(y)$ and the normalization condition: $v(k,\ldots\, k)=1$. Hence, a $k$-ary capacity represents a preference on $L$ which is increasing with the value of the attributes. 
We denote by $\mathcal{G}(L)$ the set of multichoice games defined on $L$. 

\medskip   
The derivative of $v\in\mathcal{G}(L)$ at $x\in L$ w.r.t. $i\in N$ such that $x_i<k$ is defined by
\begin{align*}
\Delta_i v(x) = v(x+1_i)-v(x).
\end{align*}
The derivative of $v\in\mathcal{G}(L)$ at $x\in L$ w.r.t. $T \subseteq N\setminus \{\emptyset\}$ such that $\forall i\in T, x_i<k$ is defined recursively as follows,
\begin{align*}
\Delta_T v(x) = \Delta_i (\Delta_{T\setminus i}v(x)).
\end{align*}
The general expression for the derivative of $v\in\mathcal{G}(L)$ is given by,
\begin{align*}
\Delta_T v(x) &= \sum_{A\subseteq T} (-1)^{t-a}v(x+1_A), \forall T \subseteq N\setminus \{\emptyset\}, \forall i\in T, x_i<k.
\end{align*}

\section{Values and interaction indices}
\label{SpreW}

\subsection{The case of classical TU-games}
\label{SubClass}

In cooperative game theory, the notion of value or power index is one of the
most important concepts. A {\it value} is a function $\phi: \mathcal{G}(2^N)
\rightarrow \reel^N$ which assigns a payoff vector to any game
$v\in\mathcal{G}(2^N)$. In MCDA, values are interpreted as importance indices
for criteria.  The {\it Shapley value} \citep{sha53} of player $i\in N$ is
given by
$$\phi_i(v) = \sum_{S\subseteq N\setminus i} \frac{(n-s-1)!s!}{n!}\big(v(S\cup i) - v(S)\big), \forall v\in\mathcal{G}(2^N).$$

The concept of interaction index, which is an extension of that of value, was
introduced axiomatically to measure the interaction phenomena among players in
cooperative game theory or criteria in multicriteria decision analysis. For a
game $v\in\mathcal{G}(2^N)$, the
{\it interaction index} of $v$ is a function $I^v: 2^N
\rightarrow \reel$ that assigns to every coalition $T\subseteq N$ its
interaction degree. 

\cite{mur93} proposed an interaction index $I(ij)$ for a pair of elements $i, j \in N$ to estimate how well $i$ and $j$ interact. \cite{gra97} defined and extended the interaction index to coalitions containing more than two players. The interaction index \citep{gra97} of a coalition $S\subseteq N$ in a game $v\in\mathcal{G}(2^N)$ is defined by
$$I^v_{Sh}(S) = \sum_{T\subseteq N\setminus S} \frac{(n-t-s)!t!}{(n-s+1)!}\sum_{K\subseteq S}(-1)^{s-k}v(K\cup T).$$
Note that when $S=\{i\}$, the interaction index coincides with the
Shapley value. 

A first axiomatization of the interaction index have been proposed by \cite{grro99}, and it is axiomatized in a way similar to the Shapley value. The following axioms have been considered by Grabisch and Roubens :
\begin{itemize}
\item Linearity axiom (L): $I^v(S)$ is linear on $\mathcal{G}(2^N)$ for every $S\subseteq N$.
\item Dummy axiom (D): For any $v\in\mathcal{G}(2^N)$, and any $i\in N$ dummy for $v$, $I^v(S\cup i)=0, \forall S\subseteq N\setminus i$.
\begin{center}
$i \in N$ is said to be dummy for $v$ if $\forall S\subseteq N\setminus i, v(S\cup i) = v(s)+v(i)$.
\end{center}

\item Symmetry axiom (S) : For any $v\in\mathcal{G}(2^N)$, any permutation $\sigma$ on $N$ and any $S\subseteq N\setminus \varnothing $, $I^v(S)=I^{\sigma v}(\sigma S).$
\item Efficiency axiom (E) :  For any $v\in\mathcal{G}(2^N)$ and any $i\in N$, $\sum_{i\in N} I^v(i)=v(N)$.
\item Recursive axiom (R1): For any $v\in\mathcal{G}(2^N)$ and any $S\subseteq N, s>1,$
$$I^v(S) = I^{v^{-j}_{\cup j}}(S\setminus j) - I^{v^{-j}}(S\setminus j), \forall j\in S,$$
where, $v^{-j}$ is the game $v$ restricted to elements in $N\setminus j$ defined by $v^{-j}(S)=v(S),\forall S\subseteq N\setminus j$, and $v^{-j}_{\cup j}$ is the game on $N\setminus j$ in the presence of $j$ defined by $v^{-j}_{\cup j}(S) = v(S\cup j) - v(S), \forall S\subseteq N\setminus j.$ 
\item Recursive axiom (R2): For any $v\in\mathcal{G}(2^N)$ and any $S\subseteq N, s>1,$
$$I^v_{Sh}(S) = I^{v_{[S]}}([S]) - \sum_{\substack{K\subseteq N\setminus S \\ K\neq \emptyset, S}}I^{v^{-K}}(S\setminus K),$$
where, $v_{[S]}$ is the game where all elements in $S$ are considered as a single element denoted $[S]$, it is defined by, for any $K\subseteq N\setminus S$:
\begin{align*}
v_{[S]}(K) &= v(K),\\
v_{[S]}(K\cup [S]) &= v(K\cup S).
\end{align*}
\end{itemize}
The axiom (R1) says that the interaction of the players in $S$ is equal to the
interaction between the criteria in $S\setminus j$ in the presence of $j$ minus
the interaction between the criteria of $S\setminus j$ in the absence of
$j$. Axiom (R2) expresses interaction of $S$ in terms of all successive interactions of subsets. The authors have shown that (R1) and (R2) are equivalent under (L), (D) and (S) axioms.

\medskip

The following theorem was shown by \cite{grro99}.
\begin{theorem}
Under axioms (L), (D), (D), (E), and ((R1) or (R2)), for all $v\in\mathcal{G}(2^N)$, 
$$I^v(S) = \sum_{T\subseteq N\setminus S} \frac{(n-t-s)!t!}{(n-s+1)!}\sum_{K\subseteq S}(-1)^{s-k}v(K\cup T), \forall S\subseteq N.$$
\end{theorem}  

\subsection{The case of games on lattices}
\label{SubGener}

\cite{grla07b} generalized the notion of interaction defined for criteria modelled by capacities, by considering functions defined on lattices. The interaction \citep{grla07b} is based on the notion of derivative of a function defined on a lattice. For this,
they introduce the following definitions:

Let $i=(0_{-j}, i_j)$ with $i_j\in L_j, j\in N$. Let $x, y \in L$ with $y=\vee_{k=1}^{n}i_k$ and $v\in\mathcal{G}(L)$. The derivative of $v$ w.r.t. $i$ at point $x\in L$ is given by: 
$$\Delta_i v(x) = v(x \vee i)-v(x),$$ 
and the derivative of $v$ w.r.t. $y$ at $x$ is given by:
$$\Delta_y v(x) = \Delta_{i_1}(\Delta_{i_2}(\ldots \Delta_{i_n} v(x) \ldots )).$$ 
The following definition has been proposed by Grabisch and Labreuche \citep{grla07b} : 
\begin{definition}
Let $J\subseteq N$, and $x=\vee_{j\in J} i_j$, with $i_j=(0_{-j}, \ell_j), \ell_j\in L_j\setminus \{0\}$.
$$I^v(x) = \sum_{y\in A(x)} \alpha^{j}_{h(y)} \Delta_x v(y),$$
where, $A(x)=\{y\in L \vert y_j=k \text{ or } 0 \text{ if } j\notin J, y_j = x_j-1 \text{ else } \}$, $h(y)$ is the number of components of $y$ to $k$ 
and $\alpha^{j}_{h(y)} = \frac{(n-j-h(y))!h(y)!}{(n-j+1)!}$.
\end{definition}

\section{Characterization of the importance index for multichoice games}\label{sec:imin}
In this section, we present the importance index (value) for multichoice games defined
by \cite{rigrla17b} together with its axiomatization.
Let $\phi$ be a value defined for any $v\in \mathcal{G}(L)$.

\medskip

\begin{quote}
\textbf{Linearity axiom (L) }: $\phi$ is linear on $\mathcal{G}(L)$, i.e., $\forall v, w \in \mathcal{G}(L), \forall\alpha\in\mathbb{R},$ 
$$\phi_i(v+\alpha w)=\phi_i(v)+\alpha\phi_i(w), \forall i\in N.$$
\end{quote}

An attribute $i\in N$ is said to be {\it null} for $v\in\mathcal{G} (L)$ if
$$v(x+1_i)=v(x), \forall x \in L, x_i<k.$$
\begin{quote}
\textbf{Null axiom (N):} If an attribute $i$ is null for $v\in\mathcal{G} (L)$, then 
$$\phi_i(v)=0.$$
\end{quote}

Let $\sigma$ be a permutation on $N$. For all $x\in L$, we denote $\sigma(x)_{\sigma(i)}=x_i$. For all $v\in\mathcal{G} (L)$,
the game $\sigma\circ v$ is defined by $\sigma\circ v (\sigma (x))=v(x)$.
\begin{quote}
\textbf{Symmetry axiom (S):} For any permutation $\sigma$ of $N$, 
$$\phi_{\sigma(i)}(\sigma\circ v)=\phi_i(v), \forall i\in N.$$
\end{quote}

\begin{quote}
\textbf{Invariance axiom (I):} Let us consider two games $v, w \in\mathcal{G} (L) $ such that, for some $i \in N$,
$$v(x+1_i)-v(x) = w(x)- w(x-1_i), \forall x\in L, x_i\notin \lbrace 0, k\rbrace$$
$$v(x_{-i}, 1_i) - v(x_{-i}, 0_i) = w(x_{-i}, k_i) - w(x_{-i}, k_i-1), \forall x_{-i}\in L_{-i},$$
then $\phi_i(v) = \phi_i(w)$.
\end{quote}

\medskip

\begin{quote}
\textbf{Efficiency axiom (E):} For all $v\in\cG(L)$,
\[
\sum_{i \in N} \phi_i(v)=\sum_{\substack{x\in L\\x_j<k}}\big(v(x+1_N)-v(x)\big) .
\]
\end{quote}

\cite{rigrla17a} have shown the following result.
\begin{theorem}
\label{THEO Imp LNISE}
Let $\phi$ be a value defined for any $v\in\mathcal{G} (L)$.
\begin{enumerate}
\item If $\phi$ fulfills (L) and (N) then there exists a family of real constants $\{b^i_x , x \in L \}$ such that
\begin{align}
\label{Imp LN}
\phi_i(v)=\sum_{\substack{x\in L\\x_i < k}} b_x^i \big(v(x+1_i)-v(x)\big), \forall i\in N.
\end{align}
\item If $\phi$ fulfills (L), (N) and (I) then
\begin{align}
\label{Imp LNI}
\phi_i(v)=\sum_{x_{-i}\in L_{-i}} b_{x_{-i}}^i \big(v(x_{-i}, k)-v(x_{-i}, 0)\big), \forall i\in N.
\end{align}
\item If $\phi$ fulfills (L), (N), (I) and (S) then 
\begin{align}
\label{Imp LNSI}
\phi_i(v)=\sum_{x_{-i}\in L_{-i}} b_{n(x_{-i})} \big(v(x_{-i}, k)-v(x_{-i}, 0)\big), \forall i\in N,
\end{align}
where $n(x_{-i})=(n_0,n_1,\ldots,n_k)$ with $n_j$ the number of components of $x_{-i}$ being equal to $j$.
\item If $\phi$ fulfills (L), (N), (I), (S) and (E) then
\begin{align}
\label{Importance}
\phi_i(v) = \sum_{x_{-i}\in L_{-i}} \frac{\big(n-\sigma(x_{-i})-1\big)!\kappa(x_{-i})!}{\big(n+\kappa(x_{-i})-\sigma(x_{-i})\big)!}\big(v(x_{-i}, k)-v(x_{-i}, 0)\big), \forall i\in N
\end{align}
\end{enumerate}
\end{theorem}

\section{Axiomatization of the interaction index}
\label{Saxiom}

In this section we intend to define axiomatically the interaction index of
multichoice games. The approach presented here is based on a recursion formula,
starting from the importance index (value) defined in Section~\ref{sec:imin}, as
in \citep{grro99}. An interaction index of the $k$-ary multichoice game $v\in
\mathcal{G}(L)$ is a function $I^v : 2^N \rightarrow \mathbb{R}$.

\bigskip
The first axiom (L) is trivially generalized for the interaction index.

\begin{quote}
\textbf{Linearity axiom (L) }: $I^v$ is linear on $\mathcal{G}(L)$, i.e., $\forall v, w \in \mathcal{G}(L), \forall\alpha\in\mathbb{R},$ 
$$  I^{v+\alpha w} = I^v + \alpha I^w .$$
\end{quote}
\begin{proposition} \label{PROP L}
Under \textbf{(L)}, for every $T\subseteq N\setminus\{\emptyset\}$, there exists real constants $ a^T_x $, for all $x\in L$, such that for every $v\in\mathcal{G}(L)$   
\begin{align}
\label{F L}
I^v(T)=\sum_{x\in L} a_x^T v(x).
\end{align}
\end{proposition}
\begin{proof}[\bf Proof ]
It is easy to check that the above formula satisfies the linearity axiom. Conversely, we consider $I^v$ satisfying \textbf{(L)}. We have $\forall v\in\mathcal{G} (L), v = \sum_{x\in L} v(x)\delta_x$. Then by \textbf{(L)},
$$I^v(T)=\sum_{x\in L} v(x) I^{\delta_x}(T), \forall T \subseteq N\setminus\{\emptyset\} $$
Setting $a^T_x=I^{\delta_x}(T), \forall x\in L, \forall T \subseteq N $, we obtain the wished result.
\end{proof}
\begin{remark}
Let $i\in N$ be a null criterion for $v\in\mathcal{G} (L)$. We have,  
$$\forall T\subseteq N, T\ni i, \Delta_T v(x) = 0, \forall x\in L, x+1_T \leq k_T.$$
$$\forall T\subseteq N, T\ni i, \Delta_j v(x)=0, \forall j\in T, \forall x\in L, x+1_j \leq k.$$
\end{remark}

\begin{quote}
\textbf{Null axiom (N):} If a criterion $i$ is null for $v\in\mathcal{G} (L)$, then  for all $T\subseteq N$ such that $T\ni i$,
$I^v(T)=0$.  
\end{quote}

\begin{proposition} \label{PROP LN}
Under axioms \textbf{(L)} and \textbf{(N)}, for every $T\subseteq N\setminus\{\emptyset\}$, there exist real constants $ b^T_x $, for all $x\in L$, with $x+1_T \leq k_T$, such that for every $v\in\mathcal{G} (L)$   
\begin{align}
\label{LN}
I^v(T)=\sum_{\substack{x\in L\\x_T < k_T}} b_x^T \Delta_T v(x).
\end{align}
\end{proposition}
To prove this result, the following lemmas are useful.
\begin{lemma} \label{Lemma1 LN}
Let $A\subseteq N$.
\begin{align*}
a_{(x_A, x_{-A})}
= \sum_{C\subseteq A} (-1)^{a-c} \sum_{\substack{\ell_C=x_C \\ \ell_{A\setminus C}=(x+1)_{A\setminus C} }}^{k_A} a_{(\ell_C, \ell_{A\setminus C}, x_{-A})}, \forall x_A \in L_A\setminus \{k_A\}. 
\end{align*}
\end{lemma}
\begin{proof}[\bf Proof ]
Let $A\subseteq N$. We proceed by recurrence on $|A|$. The relation is obviously true for $|A|=0$. Let us suppose that the relation is true for any set of at most $|A|-1$ elements, and try to show it is also true for any set of $|A|$ elements. We have , for all $x_A \in L_A\setminus \{k\}^A$,
\begin{align*}
a_{(x_A, x_{-A})} &=a_{(x_{A\setminus i}, x_i, x_{-A})}\\
&=\sum_{C\subseteq A\setminus i} (-1)^{a-c-1} \sum_{\substack{l_C=x_C \\ \ell_{A\setminus C \cup i}=(x+1)_{A\setminus C\cup i} }}^{k_{A\setminus i}} a_{(\ell_C, \ell_{A\setminus C\cup i}, x_i, x_{-A})}\\
&=\sum_{C\subseteq A\setminus i} (-1)^{a-c-1} \sum_{\substack{\ell_C=x_C \\ \ell_{A\setminus C\cup i}=(x+1)_{A\setminus C\cup i} }}^{k_{A\setminus i}} \Big(\sum_{\ell_i=x_i}^{k} a_{(\ell_C, \ell_{A\setminus C\cup i}, \ell_i, x_{-A})}-\sum_{\ell_i=x_i+1}^{k} a_{(\ell_C, \ell_{A\setminus C\cup i}, \ell_i, x_{-A})}\Big)\\
&=\sum_{C\subseteq A\setminus i} (-1)^{a-c-1} \Big(\sum_{\substack{\ell_{C\cup i}=x_{C\cup i} \\ \ell_{A\setminus C\cup i}=(x+1)_{A\setminus C\cup i} }}^{k_A}  a_{(\ell_C, \ell_{A\setminus C\cup i}, \ell_i, x_{-A})}-\sum_{\substack{\ell_C=x_C \\ \ell_{A\setminus C}=(x+1)_{A\setminus C} }}^{k_A} a_{(\ell_C, \ell_{A\setminus C\cup i}, \ell_i, x_{-A})}\Big)\\
&=\sum_{C\subseteq A\setminus i} \Big((-1)^{a-c-1} \sum_{\substack{\ell_{C\cup i}=x_{C\cup i} \\ \ell_{A\setminus C\cup i}=(x+1)_{A\setminus C\cup i} }}^{k_A} a_{(\ell_C, \ell_{A\setminus C}, x_{-A})}+(-1)^{a-c} \sum_{\substack{\ell_C=x_C \\ \ell_{A\setminus C}=(x+1)_{A\setminus C} }}^{k_A} a_{(\ell_C, \ell_{A\setminus C}, x_{-A})}\Big)\\
&=\sum_{C\subseteq A} (-1)^{a-c} \sum_{\substack{\ell_C=x_C \\ \ell_{A\setminus C}=(x+1)_{A\setminus C} }}^{k_A} a_{(\ell_C, \ell_{A\setminus C}, x_{-A})}
\end{align*}
\end{proof}
\begin{lemma} \label{Lemma2 LN}
\begin{align*}
\sum_{\substack{x \in L \\ x<k_N}} b_x \sum_{A\subseteq N}(-1)^{n-a}v(x+1_A)
&=\sum_{\substack{A \subseteq N \\ 0_A<x_A<k_A}} \sum_{\substack{B\subseteq N\setminus A \\ C\subseteq A}} (-1)^{a+b-c} b_{(x_A-1_C, 0_B, (k-1)_{N\setminus A\cup B})} v(x_A, 0_B, k_{N\setminus A\cup B})
\end{align*}
\end{lemma}
\begin{proof}[\bf Proof ]
We shall proceed by induction on $n$. For simplicity, we denote $N\setminus i$ by $S$, $(x_A, 0_B, k_{S\setminus A\cup B})$ by $x_{A,B}^S$ and $(x_A-1_C, 0_B, (k-1)_{S\setminus A\cup B})$ by $x_{A,C,B}^S$, with $C\subseteq A$. 
The relation is obviously true for $n=1$. Let us suppose that the relation is true for any set of at most $n-1$ elements, and try to show it is also true for any set of $n$ elements. We have
\begin{align*}
&\sum_{\substack{x \in L \\ x<k_N}} b_x \sum_{A\subseteq N}(-1)^{n-a}v(x+1_A)\\
&=\sum_{\substack{x \in L \\ x<k_N}} b_x \sum_{A\subseteq N\setminus i}(-1)^{n-a}\big(v(x+1_A)-v(x+1_{A\cup i})\big)\\
&=\sum_{x_i<k_i} \sum_{\substack{x_{-i} \in L_{-i} \\ x_{-i}<k_{-i}}} b_{x_{-i}, x_i} \sum_{A\subseteq S}(-1)^{s-a}\big(v(x_{-i}+1_A, x_i+1)-v(x_{-i}+1_A, x_i)\big)\\
&=\sum_{\substack{A \subseteq S \\ 0_A<x_A<k_A \\ x_i<k_i}} \sum_{\substack{B\subseteq S\setminus A \\ C\subseteq A}} (-1)^{a+b-c} b_{(x_{A,C,B}^S, x_i)}\big(v(x_{A,B}^S, x_i+1)-v(x_{A,B}^S, x_i)\big)\\
&=\sum_{\substack{A \subseteq S \\ 0_A<x_A<k_A}} \bigg[
\sum_{\substack{B\subseteq S\setminus A \\ C\subseteq A}} \Big((-1)^{a+b-c} b_{(x_{A,C,B}^S, k_i-1)} v(x_{A,B}^S, k_i)+(-1)^{a+b+1-c} b_{(x_{A,C,B}^S, 0_i)} v(x_{A,B}^S, 0_i)\Big)\\
&+\sum_{\substack{B\subseteq S\setminus A \\ C\subseteq A \\ 0<x_i<k_i}} \Big((-1)^{a+1+b-c} b_{(x_{A,C,B}^S, x_i)} v(x_{A,B}^S, x_i)+(-1)^{a+b-c} b_{(x_{A,C,B}^S, x_i-1)} v(x_{A,B}^S, x_i)\Big)\bigg]\\
&=\sum_{\substack{A \subseteq S \\ 0_A<x_A<k_A}} \bigg[\sum_{\substack{B\subseteq S\setminus A \\ C\subseteq A}} \Big((-1)^{a+b-c} b_{(x_{A,C,B}^N)} v(x_{A,B}^N)+(-1)^{a+b+1-c} b_{(x_{A,C,B\cup i}^S)} v(x_{A, B\cup i}^S)\Big)\\
&+\sum_{\substack{B\subseteq S\setminus A \\ C\subseteq A \\ 0<x_i<k_i}}  \Big((-1)^{a+1+b-c} b_{(x_{A\cup i,C,B}^S)} v(x_{A\cup i, B}^S)+(-1)^{a+b-c} b_{(x_{A\cup i,C\cup i, B}^S)} v(x_{A\cup i,B}^S)\Big)\bigg]\\
&=\sum_{\substack{A \subseteq S \\ 0_A<x_A<k_A}} \Big(\sum_{\substack{B\subseteq N\setminus A \\ C\subseteq A}} (-1)^{a+b-c} b_{(x_{A,C,B}^N)} v(x_{A, B}^N)+\sum_{\substack{B\subseteq S\setminus A \\ C\subseteq A\cup i \\ 0<x_i<k_i}} 
(-1)^{a+1+b-c} b_{(x_{A\cup i,C,B}^S)} v(x_{A\cup i,B}^S)\Big)\\
&=\sum_{\substack{A \subseteq N \\ 0_A<x_A<k_A}} \sum_{\substack{B\subseteq N\setminus A \\ C\subseteq A}} (-1)^{a+b-c} b_{(x_A-1_C, 0_B, (k-1)_{N\setminus A\cup B})}v(x_A, 0_B, k_{N\setminus A\cup B})
\end{align*}
which is the desired result.
\end{proof}
We now prove Proposition $\ref{PROP LN}$.
\begin{proof} [\bf Proof ]
It is easy to check that the formula satisfies the axioms. Conversely, we consider $I^v$ satisfying \textbf{(L)} and \textbf{(N)}. Let $v\in\mathcal{G} (L)$, and $T \in 2^N\setminus \{\emptyset\}$.\\ 
By Proposition $\ref{PROP L}$, there exists $a^T_x \in\mathbb{R}$, for all $x\in L$, such that,
$$I^v(T) = \sum_{x\in L} a_x^T v(x).$$
Then,
$$I^v(T) = \sum_{x_{-i}\in L_{-i}} \sum_{x_i \in L_i} a_{(x_{-i}, x_i)}^T v(x_{-i}, x_i).$$
Assume now that $i$ is null criterion for $v$. We have $v(x_{-i}, x_i)=v(x_{-i}, 0_i)$. Hence,
$$I^v(T) = \sum_{x_{-i}\in L_{-i}} \sum_{x_i \in L_i} a_{(x_{-i}, x_i)}^T v(x_{-i}, 0_i).$$
By \textbf{(N)}, we have, for all $i\in T$ null, and for all $x_{-i} \in L_{-i}$,
$$\sum_{x_i \in L_i} a_{(x_{-i}, x_i)}^T=0.$$
$\forall x_{-T} \in L_{-T}$, let $A \subseteq T, B\subseteq T\setminus A$, and set
$$b_{(x_A, 0_B, (k-1)_{T\setminus A\cup B}, x_{-T})}^T = (-1)^b \sum_{\ell_A=(x+1)_A}^{k_A} a_{(\ell_A, 0_B, k_{T\setminus A\cup B}, x_{-T})}^T,  \forall x_A\in L_A\setminus \{0, k\}^A.$$
Then, we have, $\forall x_{-T}\in L_{-T}, \forall A\subseteq T, \forall B\subseteq T\setminus A, \forall x_A\in L_A\setminus \{0, k\}^A,$
\begin{align*}
a_{(x_A, 0_B, k_{T\setminus A\cup B}, x_{-T})}^T
&= \sum_{C\subseteq A} (-1)^{a-c} \sum_{\substack{\ell_C=x_C \\ \ell_{A\setminus C}=(x+1)_{A\setminus C} }}^{k_A} a_{(\ell_C, \ell_{A\setminus C}, 0_B, k_{T\setminus A\cup B}, x_{-T})}^T \text{ (using Lemma }\ref{Lemma1 LN})\\
&= \sum_{C\subseteq A} (-1)^{a-c} \sum_{\substack{\ell_C=(x_C-1_C)+1_C \\ \ell_{A\setminus C}=x_{A\setminus C}+1_{A\setminus C} }}^{k_A} a_{(\ell_C, \ell_{A\setminus C}, 0_B, k_{T\setminus A\cup B}, x_{-T})}^T \\
&= (-1)^b\sum_{C\subseteq A} (-1)^{a-c} (-1)^b \sum_{y=((x-1_C)+1)_A}^{k_A} a_{(y, 0_B, k_{T\setminus A\cup B}, x_{-T})}^T\\
&= (-1)^b\sum_{C\subseteq A} (-1)^{a-c} b_{(x_A-1_C, 0_B, (k-1)_{T\setminus A\cup B}, x_{-T})}^T\\
\end{align*}
Therefore, it suffices to replace the values of $a_x^T$ in the formula $(\ref{F L})$, and then the result is established.
\begin{align*}
I^v(T) &= \sum_{x_{-T}\in L_{-T}}\sum_{x_T\in L_T} a_{(x_T, x_{-T})}^T v(x_T, x_{-T})\\
&=\sum_{x_{-T}\in L_{-T}} \sum_{\substack{A \subseteq T \\ 0_A<x_A<k_A}} \sum_{B\subseteq T\setminus A} a_{(x_A, 0_B, k_{T\setminus A\cup B}, x_{-T})}^T v(x_A, 0_B, k_{T\setminus A\cup B}, x_{-T})\\
&=\sum_{x_{-T}\in L_{-T}} \sum_{\substack{A \subseteq T \\ 0_A<x_A<k_A}} \sum_{B\subseteq T\setminus A} (-1)^b \sum_{C\subseteq A}(-1)^{a-c}b_{(x_A-1_C, 0_B, (k-1)_{T\setminus A\cup B}, x_{-T})}^T v(x_A, 0_B, k_{T\setminus A\cup B}, x_{-T})\\
&=\sum_{x_{-T}\in L_{-T}} \sum_{\substack{x_T \in L_T \\ x_T<k_T}} b_x^T \sum_{A\subseteq T}(-1)^{t-a}v(x+1_A) \text{ (using Lemma }\ref{Lemma2 LN})\\
&=\sum_{\substack{x\in L \\ x_T<k_T}} b_x^T \sum_{A\subseteq T}(-1)^{t-a}v(x+1_A).
\end{align*}
\end{proof}
\begin{quote}
\textbf{Invariance axiom (I):} Let us consider two functions $v, w \in\mathcal{G} (L) $ such that, for all $i \in N$,
$$v(x+1_i)-v(x) = w(x)- w(x-1_i), \forall x\in L, x_i\notin \lbrace 0, k\rbrace$$
$$v(x_{-i}, 1_i) - v(x_{-i}, 0_i) = w(x_{-i}, k_i) - w(x_{-i}, k_i-1), \forall x_{-i}\in L_{-i}.$$
Then $I^v(T\cup i) = I^w(T\cup i), \forall T \subseteq N\setminus i$.
\end{quote}

\begin{proposition} \label{PROP LNI}
Under axioms \textbf{(L)}, \textbf{(N)} and \textbf{(I)}, $\forall v\in\mathcal{G} (L), \forall i\in N$,  
$$I^v(T)=\sum_{x_{-T}\in L_{-T}} b_{x_{-T}}^T \sum_{S\subseteq T}(-1)^{t-s} v(0_S, k_{T\setminus S}, x_{-T}).$$
\end{proposition}
\begin{proof}[\bf Proof ]
It is easy to check that the above formula satisfies the axioms.  Conversely, we consider $I^v$ satisfying \textbf{(L)}, \textbf{(N)} and \textbf{(I)}. Let $v, w \in\mathcal{G} (L)$, and $T \subseteq N$. By Proposition \ref{PROP LN} and the axiom \textbf{(I)}, we have, for any $i\in T$
\begin{align*}
I^v(T)&=\sum_{\substack{x\in L\\x_T < k_T}} b_x^T \Delta_T v(x)\\
&=\sum_{\substack{x_{-i}\in L_{-i}\\x_{T\setminus i} < k_{T\setminus i}}} \Big(b_{(0_i,x_{-i})}^T \Delta_{T\setminus i} \Delta_i v(0_i, x_{-i})+\sum_{\substack{x_i\in L_i \\ x_i \notin \{0, k\}}}b_x^T \Delta_{T\setminus i} \Delta_i v(x)\Big)\\
&=\sum_{\substack{x_{-i}\in L_{-i}\\x_{T\setminus i} < k_{T\setminus i}}} \Big(b_{(0_i,x_{-i})}^T \Delta_{T\setminus i} \Delta_i w\big((k-1)_i,x_{-i}\big)+\sum_{\substack{x_i\in L_i \\ x_i \notin \{0, k\}}}b_x^T \Delta_{T\setminus i} \Delta_i w(x-1_i)\Big)\\
&=\sum_{\substack{x_{-i}\in L_{-i}\\x_{T\setminus i} < k_{T\setminus i}}} \Big(b_{(0_i,x_{-i})}^T \Delta_{T\setminus i} \Delta_i w\big((k-1)_i, x_{-i}\big)+\sum_{\substack{x_i\in L_i \\ x_i <k-1}}b_{x_i+1_i,x_{-i}}^T \Delta_{T\setminus i} \Delta_i w(x)\Big),
\end{align*}
and,
\begin{align*}
I^w(T)&=
\sum_{\substack{x_{-i}\in L_{-i}\\x_{T\setminus i} < k_{T\setminus i}}} \Big(b_{\big((k-1)_i, x_{-i} \big)}^T \Delta_{T\setminus i} \Delta_i w\big((k-1)_i , x_{-i} \big)+\sum_{\substack{x_i\in L_i \\ x_i < k-1}}b_x^T \Delta_{T\setminus i} \Delta_i w(x)\Big),
\end{align*}
then, $b_{x_i, x_{-i}}^T = b_{x_i+1_i, x_{-i}}^T, \forall x_{-i} \in L_{-i}, \forall x_i \in L_i \setminus \{k, k-1\}$ and any $i\in T$. Hence, $b_{x_T, x_{-T}}^T = b_{(x+1)_T, x_{-T}}^T$, for all $x_{-T}\in L_{-T}$ and for all $x_T\in L_T$ such that $x_T<k_T$. 

We conclude that the coefficient $b_{x_T, x_{-T}}^T$ does not depend on $x_T$. 

We set thus $b_{x_{-T}}^T := b_{x_T, x_{-T}}^T$. Hence, for any $v\in\mathcal{G}(L)$, and for any $T \subseteq N$, we have,  
\begin{align*}
I^v(T) &= \sum_{\substack{x\in L\\x_T < k_T}} b_x^T \Delta_T v(x)\\
&= \sum_{x_{-T}\in L_{-T}} b_{x_{-T}}^T \sum_{\substack{x_T \in L_T \\ x_T < k_T}} \Delta_T v(x)\\
&= \sum_{x_{-T}\in L_{-T}} b_{x_{-T}}^T \sum_{S\subseteq T}(-1)^{t-s} v(0_{T\setminus S}, k_S, x_{-T})
\end{align*}
\end{proof}

\bigskip
We introduce the Symmetry axiom. 
\begin{quote}
\textbf{Symmetry axiom (S):} For all $v\in\mathcal{G}(L)$, for all permutation $\sigma$ on $N$,
$$I^{\sigma \circ v}(\sigma (T)) = I^v(T), \forall T\subseteq N.$$  
\end{quote}
\begin{proposition} \label{PROP LNS}
Under axioms \textbf{(L)}, \textbf{(N)}, \textbf{(S)}, $\forall v \in \mathcal{G}(L), \forall T\subseteq N$,  
\begin{align}
\label{L}
I^v(T)=\sum_{\substack {x\in L \\ x_T <k_T}} b_{x_T; n_0, n_1, \ldots, n_k} \Delta_T v(x),
\end{align}
where $b_{x_T; n_0, n_1, \ldots, n_k}\in\mathbb{R}$, and $n_j=|\{\ell\in N\setminus T, x_\ell = j\}|$
\end{proposition}
\begin{proof}[\bf Proof ]
Let $v\in\mathcal{G} (L)$ and let $\sigma$ be a permutation on $N$. For every $x\in L$, we put $y=\sigma^{-1}(x)$. From Proposition $\ref{PROP LN}$, we have $\forall T \subseteq N$
\begin{align*}
I^v(T)=\sum_{\substack{y\in L\\y_T < k_T}} b_y^T \Delta_T v(y),
\end{align*}
and,
\begin{align*}
I^{\sigma \circ v}(\sigma(T))
&=\sum_{\substack{x\in L\\x_{\sigma(T)} < k_{\sigma(T)}}} b_x^{\sigma(T)} \Delta_{\sigma(T)} \sigma \circ v(x)\\
&=\sum_{\substack{y\in L\\y_T < k_T}} b_{\sigma(y)}^{\sigma(T)} \Delta_T v(y).
\end{align*}
Then, from the symmetry axiom, we have for all $y\in L$ such that $y_T<k_T$: $b_{\sigma(y)}^{\sigma(T)} = b_y^T$.

For every $y \in L$ such that $y_T<k$, we can write,
$$b_{(y_T; y_{-T})}^T=b_y^T=b_{\sigma(y)}^{\sigma(T)}=b_{(\sigma(y)_{\sigma(T)}; \sigma(y)_{-\sigma(T)})}^{\sigma(T)} = b_{(y_T; \sigma(y)_{-\sigma(T)})}^{\sigma(T)}$$\\
Assuming that $\sigma(T)=T$, then,
$$b_{(y_T; y_{-T})}^T = b_{(y_T; \sigma(y)_{-\sigma(T)})}^{T}$$
For a fixed $T$, $b_{(y_T; \sigma(y)_{-\sigma(T)})}^{T}$ depends only on $n(y_{-T})$, with $n(y_{-T})=\lbrace n_0(y_{-T}), n_1(y_{-T}),\ldots, n_k(y_{-T}) \rbrace$, and $n_j(y_{-T})=|\lbrace \ell\in N\setminus T \vert y_\ell=j\rbrace|.$
$$p_{y_T; y_{-T}}^T = p_{y_T; n(y_{-T})}^{T}$$
Suppose now that $\sigma(T)=S$ (with $S \neq T$), and $\sigma(\ell)=\ell, \forall \ell\in N\setminus S\cup T$, then,  
\begin{align*}
b_{(y_T; n(y_{-T}))}^T &= b_{(y_T; n(\sigma(y)_{-\sigma(T)}))}^{\sigma(T)}\\
 &= b_{(y_T; n(y_{-T}))}^{\sigma(T)}
\end{align*}
we can conclude that the value $b_{y_T; n(y_{-T})}^{T}$ does not depend on the exponent $T$. We denote by $b_{y_T; n(y_{-T})}$ this value.
\end{proof}

\begin{proposition} \label{PROP LNIS}
Under axioms \textbf{(L)}, \textbf{(N)}, \textbf{(I)} and \textbf{(S)}, for any $v \in \mathcal{G}(L)$, $\forall T\subseteq N$,  
\begin{align}
\label{LNIS}
I^v(T)=\sum_{x_{-T}\in L_{-T}} b_{n(x_{-T})} \sum_{S\subseteq T}(-1)^{t-s} v(0_{T\setminus S}, k_S, x_{-T}),
\end{align}
where $b_{n(x_{-T})}\in\mathbb{R}$, $n(x_{-T})=(n_0, n_1, \ldots, n_k)$ and $n_j=|\{\ell\in N\setminus T, x_\ell = j\}|$
\end{proposition}

\begin{quote}
\textbf{Efficiency axiom (E):} For all $v\in\cG(L)$, 
\[
\sum_{i \in N} I^v (i)=\sum_{\substack{x\in L\\x_j<k}}\big(v(x+1_N)-v(x)\big) .
\]
\end{quote}

We introduce now the Recursivity axiom which is the exact counterpart of the one for classical games in \citep{grro99}. For this, we introduce the following definitions:

Let $v$ be a multichoice game in $\mathcal{G}(L)$ and $S \subseteq N$. 
We introduce the restricted multichoice game $v^{-S}$ of $v$, which is defined on $N\setminus S$ as follows
$$v^{-S}(x_{-S}) = v(x_{-S},0_S), \forall x_{-S}\in L_{-S}.$$
The restriction of $v$ to $i\in N$ in the presence of $i$ denoted by $v_{i}^{-i}$  is the multichoice game on $L_{-i}$ defined by
$$v_{i}^{-i}(x_{-i}) = v(x_{-i}, k_i)-v(0_{-i}, k_i), \forall x_{-i}\in L_{-i}.$$
\begin{quote}
\textbf{Recursivity axiom (R):} For any $v\in\mathcal{G}(L)$,
$$I^v (T) = I^{v_{i}^{-i}}(T\setminus i) - I^{v^{-i}}(T\setminus i), \forall T\subseteq N\setminus \{\emptyset\}, \forall i\in T.$$  
\end{quote}
\begin{lemma}
\label{Lemma R}
Under axioms \textbf{(L)}, \textbf{(N)}, \textbf{(I)} \textbf{(S)} and \textbf{(R)}, for any $v \in \mathcal{G}(L)$, $\forall T\subseteq N \setminus \{\emptyset\}$,
\begin{align}
\label{Formule R}
I^v(T)=\sum_{\substack{A\subseteq T \\ A\neq \emptyset}} (-1)^{t-a} I^{v^{(-T)\cup [A]}_{[A]}}([A]),
\end{align}
with $v^{(-T)\cup [A]}_{[A]}$ is the reduced multichoice game of $v$ to $T$ with respect to $A$ defined on the set $\{0,\ldots, k\}^{(N\setminus T)\cup [A]}$ as follows:
$$v_{[A]}^{(-T) \cup [A]}(x_{-T}, \ell_{[A]}) = v(x_{-T}, \ell_A, 0_{T\setminus A}), \ell \in \{0, \ldots, k \}.$$
\end{lemma}
\begin{proof} [\bf Proof ]
We suppose that the axioms \textbf{(L)}, \textbf{(N)}, \textbf{(I)}, \textbf{(S)} and \textbf{(R)} are satisfied. We proceed by induction on $|T|$. The formula is true for $|T|=1$. Let us assume it is true up to $|T|=t-1$, and try to prove it for $t$ elements. 
By induction assumption we have, for any $v\in\mathcal{G}(L)$, and $i\in T$,
\begin{align*}
I^{v^{-i}}(T\setminus i)=\sum_{\substack{A\subseteq T\setminus i \\ A\neq \emptyset}} (-1)^{t-a-1} I^{v^{(-T)\cup [A]}_{[A]}}([A]),
\end{align*}
\begin{align*}
I^{v_{i}^{-i}}(T\setminus i)=\sum_{\substack{A\subseteq T\setminus i \\ A\neq \emptyset}} (-1)^{t-a-1} I^{v^{(-T)\cup [A]}_{[A], i}}([A]).
\end{align*}
with $v_{[A], i}^{(-T)\cup [A]}(x_{-T}, \ell_{[A]}) = v(x_{-T}, \ell_A, k_i, 0_{T\setminus A\cup i})-v(0_{-i}, k_i), \ell \in \{0, \ldots, k \}$. 

\bigskip
Let $A \subseteq T\setminus i$ such that $A\neq \emptyset$. From Proposition \ref{PROP LNIS},
\begin{align*}
I^{v^{(-T)\cup [A]}_{[A], i}}([A])
&=\sum_{x_{-T}\in L_{-T}} b_{n(x_{-T})} \Big(v^{(-T)\cup [A]}_{[A], i}(x_{-T}, k_{[A]})-v^{(-T)\cup [A]}_{[A], i}(x_{-T}, 0_{[A]})\Big)\\
&=\sum_{x_{-T}\in L_{-T}} b_{n(x_{-T})} \Big(v(x_{-T}, k_{A\cup i}, 0_{T\setminus A\cup i})-v(x_{-T}, k_i, 0_{T\setminus i})\Big) \\
&= \sum_{x_{-T}\in L_{-T}} b_{n(x_{-T})} \Big(v(x_{-T}, k_{A\cup i}, 0_{T\setminus A\cup i})-v(x_{-T}, 0_T)\Big)\\
&-\sum_{x_{-T}\in L_{-T}} b_{n(x_{-T})} \Big(v(x_{-T}, k_i, 0_{T\setminus i})-v(x_{-T}, 0_T)\Big)\\
&= \sum_{x_{-T}\in L_{-T}} b_{n(x_{-T})} \Big(v^{(-T)\cup [A\cup i]}_{[A\cup i]}(x_{-T}, k_{[A\cup i]})-v^{(-T)\cup [A\cup i]}_{[A\cup i]}(x_{-T}, 0_{[A\cup i]})\Big)\\
&-\sum_{x_{-T}\in L_{-T}} b_{n(x_{-T})} \Big(v^{(-T)\cup i}(x_{-T}, k_i)-v^{(-T)\cup i}(x_{-T}, 0_i)\Big)\\
&=I^{v^{(-T)\cup [A\cup i]}}_{[A\cup i]}([A\cup i])-I^{v^{(-T)\cup i}}(i).
\end{align*}
By \textbf{(R)}, we have
\begin{align*}
I^v (T) &= I^{v_{i}^{-i}}(T\setminus i) - I^{v^{-i}}(T\setminus i)\\
&=\sum_{\substack{A\subseteq T\setminus i \\ A\neq \emptyset}} (-1)^{t-a-1} I^{v^{(-T)\cup [A]}_{[A], i}}([A])-\sum_{\substack{A\subseteq T\setminus i \\ A\neq \emptyset}} (-1)^{t-a-1} I^{v^{(-T)\cup [A]}_{[A]}}([A])\\
&=\sum_{\substack{A\subseteq T\setminus i \\ A\neq \emptyset}} (-1)^{t-a-1} \Big(I^{v^{(-T)\cup [A\cup i]}}_{[A\cup i]}([A\cup i])-I^{v^{(-T)\cup i}}(i)\Big)-\sum_{\substack{A\subseteq T\setminus i \\ A\neq \emptyset}} (-1)^{t-a-1} I^{v^{(-T)\cup [A]}_{[A]}}([A])\\
&=\sum_{\substack{A\subseteq T\setminus i \\ A\neq \emptyset}} (-1)^{t-a-1} \Big(I^{v^{(-T)\cup [A\cup i]}}_{[A\cup i]}([A\cup i])-I^{v^{(-T)\cup [A]}_{[A]}}([A])\Big)-I^{v^{(-T)\cup i}}(i) \sum_{\substack{A\subseteq T\setminus i \\ A\neq \emptyset}} (-1)^{t-a-1} \\
&=\sum_{\substack{A\subseteq T\setminus i \\ A\neq \emptyset}} (-1)^{t-a-1} \Big(I^{v^{(-T)\cup [A\cup i]}}_{[A\cup i]}([A\cup i])-I^{v^{(-T)\cup [A]}_{[A]}}([A])\Big)+(-1)^{t-1}I^{v^{(-T)\cup [i]}_{[i]}}([i]) \\
&=\sum_{\substack{A\subseteq T \\ A\neq \emptyset}} (-1)^{t-a} I^{v^{(-T)\cup [A]}_{[A]}}([A])
\end{align*}
\end{proof}

\begin{theorem}\label{THEO INT Sha}
Under axioms \textbf{(L)}, \textbf{(N)}, \textbf{(I)}, \textbf{(S)}, \textbf{(E)} and \textbf{(R)}, $\forall v\in\mathcal{G} (L), \forall T \subseteq N\setminus \varnothing,$
$$I^v(T) = I^v_{\mathrm{s}}(T) := \sum_{x_{-T}\in L_{-T}} \frac{(n-s(x_{-T})-t)!k(x_{-T})!}{(n-s(x_{-T})+k(x_{-T})-t+1)!} \sum_{A\subseteq T}(-1)^{t-a}v(0_{T\setminus A}, k_A, x_{-T}).$$
\end{theorem} 

\medskip

\begin{proof} [\bf Proof ]
Let $v\in\mathcal{G} (L)$, and $T\in N \setminus \{\emptyset\}$. By axioms \textbf{(L)}, \textbf{(N)}, \textbf{(I)}, \textbf{(S)} and \textbf{(E)}, we have
\begin{align*} 
I^{v^{(-T)\cup [A]}_{[A]}}([A])
&=\sum_{x_{T}\in L_{-T}} b_{n(x_{-T})}\big(v^{(-T)\cup [A]}_{[A]}(x_{-T}, k_{[A]})-v^{(-T)\cup [A]}_{[A]}(x_{-T}, 0_{[A]})\big), 
\end{align*}
with $b_{n(x_{-T})} =\displaystyle\frac{\big(n-t-s(x_{-T})\big)!k(x_{-T})!}{\big(n-t+1+k(x_{-T})-s(x_{-T})\big)!}$. 

\bigskip

By Lemma (\ref{Lemma R}), we have
\begin{align*}
I^v(T)
&=\sum_{\substack{A\subseteq T \\ A\neq \emptyset}} (-1)^{t-a} I^{v^{(-T)\cup [A]}_{[A]}}([A])\\
&=\sum_{\substack{A\subseteq T \\ A\neq \emptyset}} (-1)^{t-a} \sum_{x_{-T}\in L_{-T}} b_{n(x_{-T})} \big(v(x_{-T}, k_A, 0_{T \setminus A})-v(x_{-T}, 0_T)\big)\\
&=\sum_{x_{-T}\in L_{-T}} b_{n(x_{-T})} \sum_{\substack{A\subseteq T \\ A\neq \emptyset}} (-1)^{t-a} \big(v(x_{-T}, k_A, 0_{T \setminus A})-v(x_{-T}, 0_T)\big)\\
&=\sum_{x_T\in L_{-T}} b_{n(x_{-T})} \sum_{A\subseteq T} (-1)^{t-a} v(k_A, 0_{T \setminus A}, x_{-T}).
\end{align*}
\end{proof}

\section{Interaction indices for the Choquet integral}\label{sec:Cho}
We propose in this section an interpretation of the interaction in continuous spaces, that is, after extending $v$ to the continuous domain $[0, k]^N$. The most usual extension of $v$ on $[0,k]^N$ is the Choquet integral with respect to $k$-ary capacities \citep{grla03b}. 

\medskip

Let $z\in[0, k]^N$, and $q\in L$ such that $q=\lfloor z \rfloor$ (the floor integer part of $z$). The Choquet integral w.r.t. a $k$-ary capacity $v$ at point $z$ is defined by
$$\mathcal{C}_{v}(z) = v(q) + \mathit{C}_{\mu_q} (z-q),$$
where $\mu_q$ is a capacity given by 
$$\mu_q(A)=v((q+1)_A, q_{-A}) - v(q), \forall A \subset N.$$
\begin{proposition}
\label{PROP InteShaSum}
For every $v\in\mathcal{G} (L)$,
$$I^v_{\mathrm{s}}(T) = \sum_{x\in \{0,\ldots, k-1\}^N} I^{\mu_x}_{Sh}(T), \forall T \subseteq N\setminus \{\emptyset\}.$$
\end{proposition}
To prove this result, the following combinatorial result is useful.
\begin{lemma}
$$\sum_{S\subseteq [A, B]} \frac{(n-s-1)!s!}{n!} = \frac{(n-b-1)!a!}{(n-b+a)!}, \forall A, B\subseteq N,A\subseteq B$$
\end{lemma}
\begin{proof} [\bf Proof ]
Let $A, B\subseteq N$, such that $A\subseteq B,$
\begin{align*}
\sum_{S\subseteq [A, B]} \frac{(n-s-1)!s!}{n!} 
&= \sum_{S\subseteq [\varnothing, B\setminus A]} \frac{(n-s-a-1)!(s+a)!}{n!}\\
&= \sum_{s=0}^{b-a} \binom {b-a} s \frac{(n-s-a-1)!(s+a)!}{n!}\\
&= \sum_{s=0}^{b-a} \binom {b-a} s \int_0^1 x^{n-s-a-1} (1-x)^{s+a}dx\\
&= \int_0^1 x^{n-b-1} (1-x)^a \sum_{s=0}^{b-a} \binom {b-a} s  x^{b-a-s} (1-x)^s dx\\
&= \int_0^1 x^{n-b-1} (1-x)^adx\\
&= \frac{(n-b-1)!a!}{(n-b+a)!}
\end{align*}
\end{proof}
We now prove Proposition \ref{PROP InteShaSum}.
\begin{proof}[\bf Proof ]
Let $T \subseteq N \setminus \{\emptyset\}$. 
\begin{align*}  
\sum_{x\in \{0,\ldots, k-1\}^N} I^{\mu_x}_{Sh}(T)
&= \sum_{x\in \{0,\ldots, k-1\}^N} \sum_{S\subseteq N\setminus T} \frac{(n-s-t)!s!}{(n-t+1)!}\Delta_T \mu_x (S)\\
&= \sum_{x\in \{0,\ldots, k-1\}^N} \sum_{S\subseteq N\setminus T} \frac{(n-s-t)!s!}{(n-t+1)!}\Delta_T v(x+1_S)\\
&= \sum_{\substack{z\in L \\ z_T<k_T}} \Delta_T v(z) \sum_{\substack{S\subseteq N\setminus T \\ \forall j\in S, z_j>0 \\ \forall j\in N\setminus S, z_j<k}} \frac{(n-s-t)!s!}{(n-t+1)!}\\
&= \sum_{\substack{z\in L \\ z_T<k_T}} \Delta_T v(z) \sum_{\substack{S\subseteq N\setminus T \cap S(z_{-T}) \\ S\supseteq K(z_{-T})}} \frac{(n-s-t)!s!}{(n-t+1)!}\\
&= \sum_{\substack{z\in L \\ z_T<k_T}} \Delta_T v(z) \sum_{\substack{S\subseteq S(z_{-T}) \\ S\supseteq K(z_{-T})}} \frac{(n-s-t)!s!}{(n-t+1)!}\\
&= \sum_{\substack{z\in L \\ z_T<k_T}} \frac{(n-s(z_{-T})-t)!k(z_{-T})!}{(n-s(z_{-T})+k(z_{-T})-t+1)!} \Delta_T v(z).
\end{align*}
\end{proof}
The interaction index on continuous domain $[0, k]^N$ takes the form of the total over the domain $\{0,\ldots, k-1\}^N$ of the classical interaction index, it means that for each elementary cell in the grid L, the interaction index corresponds to the usual interaction index.

\medskip

\begin{theorem}
Let $v$ a $k$-ary capacity.
$$I^v_{\mathrm{s}} (T) = \int_{[0, k]^n} \frac{\partial^{|T|} \mathcal{C}_{v}}{\partial z_T}(z) \: dz\\, \forall T \subseteq N.$$ 
\end{theorem}
\begin{proof} [\bf Proof ]
Let $v$ a $k$-ary capacity. For every $T \subseteq N$. $\forall x\in \{0,\ldots, k-1\}^N$, we have,
\begin{align*}  
I^v_{\mathrm{s}} (T)
&=\sum_{x\in \{0,\ldots, k-1\}^N} I^{\mu_x}_{Sh}(T)\\
&= \sum_{x\in \{0,\ldots, k-1\}^N} \int_{[0, 1]^n} \frac{\partial^{t} \mathcal{C}_{\mu_x}}{\partial z_T}(z) \: dz\\
&= \sum_{x\in \{0,\ldots, k-1\}^N} \int_{[x, x+1_N]} \frac{\partial^{t} \mathcal{C}_{\mu_x}}{\partial z_T}(z-x) \: dz\\
&= \sum_{x\in \{0,\ldots, k-1\}^N} \int_{[x, x+1_N]} \frac{\partial^{t} \mathcal{C}_{v}}{\partial z_T}(z) \: dz\\
&=\int_{[0, k]^n} \frac{\partial^{t} \mathcal{C}_{v}}{\partial z_T}(z) \: dz. 
\end{align*}
\end{proof}
The interaction index on continous domain appears as the mean of relative amplitude of the range of $\mathcal{C}_{v}$ w.r.t. $T$, when the remaining variables take uniformly random values. The partial derivative is the local interaction of $\mathcal{C}_{v}$ at point $z$.

\bibliographystyle{plainnat}
\bibliography{References}

\begin{thebibliography}{13}
\providecommand{\natexlab}[1]{#1}
\providecommand{\url}[1]{\texttt{#1}}
\expandafter\ifx\csname urlstyle\endcsname\relax
  \providecommand{\doi}[1]{doi: #1}\else
  \providecommand{\doi}{doi: \begingroup \urlstyle{rm}\Url}\fi

\bibitem[Choquet(1953)]{cho53}
G.~Choquet.
\newblock Theory of capacities.
\newblock \emph{Annales de l'institut Fourier}, 5:\penalty0 131--295, 1953.

\bibitem[Grabisch(1997)]{gra97}
M.~Grabisch.
\newblock $k$-order additive discrete fuzzy measures and their representation.
\newblock \emph{Fuzzy Sets and Systems}, 92\penalty0 (2):\penalty0 167--189,
  1997.

\bibitem[Grabisch and Labreuche(2003)]{grla03b}
M.~Grabisch and Ch. Labreuche.
\newblock Capacities on lattices and k-ary capacities.
\newblock In \emph{In 3d Int, Conf. of the European Soc. for Fuzzy Logic and
  Technology (EUSFLAT)}, pages 304--307, Zittau, Germany, September 10-12 2003.

\bibitem[Grabisch and Labreuche(2007)]{grla07b}
M.~Grabisch and Ch. Labreuche.
\newblock Derivative of functions over lattices as a basis for the notion of
  interaction between attributes.
\newblock \emph{Annals of Mathematics and Artificial Intelligence}, 49\penalty0
  (1-4):\penalty0 151--170, 2007.

\bibitem[Grabisch and Labreuche(2010)]{grla10}
M.~Grabisch and Ch. Labreuche.
\newblock A decade of application of the {C}hoquet and {S}ugeno integrals in
  multi-criteria decision aid.
\newblock \emph{Annals of Operations Research}, 175:\penalty0 247--286, 2010.

\bibitem[Grabisch and Roubens(1999)]{grro99}
M.~Grabisch and M.~Roubens.
\newblock An axiomatic approach to the concept of interaction among players in
  cooperative games.
\newblock \emph{International Journal of Game Theory}, 28\penalty0
  (4):\penalty0 547--565, 1999.

\bibitem[Hsiao and Raghavan(1993)]{hsra93}
C.~R. Hsiao and T.~E.~S. Raghavan.
\newblock Shapley value for multi-choice cooperative games, {I}.
\newblock \emph{Games and Economic Behavior}, 5:\penalty0 240–256, 1993.

\bibitem[Keeney and Raiffa(1976)]{kera76}
R.~L. Keeney and H.~Raiffa.
\newblock \emph{Decision with Multiple Objectives}.
\newblock Wiley, New York, 1976.

\bibitem[Murofushi and Soneda(1993)]{muso93}
T.~Murofushi and S.~Soneda.
\newblock Techniques for reading fuzzy measures (iii): interaction index.
\newblock In \emph{9th fuzzy system symposium}, pages 693--696. Sapporo, Japan,
  1993.

\bibitem[Ridaoui et~al.(2017{\natexlab{a}})Ridaoui, Grabisch, and
  Labreuche]{rigrla17a}
M.~Ridaoui, M.~Grabisch, and Ch. Labreuche.
\newblock An alternative view of importance indices for multichoice games.
\newblock In \emph{International Conference on Algorithmic Decision Theory},
  pages 81--92, 2017{\natexlab{a}}.

\bibitem[Ridaoui et~al.(2017{\natexlab{b}})Ridaoui, Grabisch, and
  Labreuche]{rigrla17b}
M.~Ridaoui, M.~Grabisch, and Ch. Labreuche.
\newblock Axiomatization of an importance index for generalized additive
  independence models.
\newblock In \emph{Symbolic and Quantitative Approaches to Reasoning with
  Uncertainty}, pages 340--350, 2017{\natexlab{b}}.

\bibitem[Shapley(1953)]{sha53}
L.~S. Shapley.
\newblock A value for $n$-person games.
\newblock In H.~W. Kuhn and A.~W. Tucker, editors, \emph{Contributions to the
  Theory of Games, Vol. II}, number~28 in Annals of Mathematics Studies, pages
  307--317. Princeton University Press, 1953.

\bibitem[Sugeno(1974)]{sug74}
M.~Sugeno.
\newblock \emph{Theory of fuzzy integrals and its applications}.
\newblock PhD thesis, Tokyo Institute of Technology, 1974.

\end{thebibliography}

\end{document}